\documentclass[preprint,12pt]{elsarticle}

\usepackage{amssymb}
\usepackage{mathrsfs}
\usepackage{mathrsfs}
\usepackage{mathrsfs}
\usepackage{amsfonts}
\usepackage{amsfonts}
\usepackage{amsfonts}
\usepackage{mathrsfs}
\usepackage{amsmath}
\usepackage{graphicx}
\usepackage{multirow}
\usepackage{caption2}
\usepackage{float}
\usepackage{booktabs}
\usepackage{algorithm}
\usepackage{algorithmic}
\usepackage{subfigure}

\newtheorem{theorem}{Theorem}[section]

\numberwithin{equation}{section}

\newtheorem{corollary}[theorem]{Corollary}

\newtheorem{lemma}[theorem]{Lemma}

\newtheorem{proposition}[theorem]{Proposition}

\newenvironment{proof}[1][Proof]{\textbf{#1. }}{\ \rule{0.5em}{0.5em}}%

\graphicspath{{figures/}}
\journal{}

\begin{document}

\begin{frontmatter}

\title{Approximation smooth and sparse functions
 by deep neural networks  without saturation  \tnoteref{t1}} \tnotetext[t1]{The research was
supported by the National Natural Science Foundation of China
(Grant Nos. 61806162, 11501496 and 11701443)}
\author{Xia Liu$^{1^*}$}\cortext[*]{Corresponding author: liuxia1232007@163.com}


\address{1. School of Sciences, Xi'an University of Technology,
Xi'an 710048, China}


\begin{abstract}
Constructing neural networks for function approximation is a classical and
longstanding topic in approximation theory.
In this paper, we aim at constructing deep neural networks (deep nets for short) with three hidden layers
to approximate smooth and sparse functions.
In particular, we prove that the constructed  deep nets can reach the optimal approximation rate in approximating
both smooth and sparse functions with controllable magnitude of free parameters.
Since the saturation that describes the bottleneck of approximate is an insurmountable problem of constructive neural networks,
we also prove that deepening the neural network with only one more hidden layer can avoid the saturation.
The obtained results underlie advantages of deep nets and provide
theoretical explanations for deep learning.

\end{abstract}

\begin{keyword}
Approximation theory, deep learning, deep neural networks,
localized approximation, sparse approximation.

\end{keyword}
\end{frontmatter}

\section{Introduction}

Machine learning \cite{Bishop2006} is a key sub-field of artificial intelligence (AI)
which abounds in sciences, engineering, medicine, computational finance
and so on. Neural network \cite{Maiorov2006,Mhaskar,Hagan1996,Lin2016} is an eternal topic
of machine learning that makes machine
learning no longer just like a machine to execute commands,
but makes machine learning have the ability to draw inferences about other cases
from one instance.
Deep learning \cite{Goodfellow2016,Hinton2006} is a new active area of machine learning research
based on deep structured learning model with appropriate algorithms, and is
acclaimed as a magical approach to deal with massive data. Indeed, neural networks with
more than one hidden layer are one of the most typical deep
structured models in deep learning \cite{Goodfellow2016}.
In current literature \cite{Lin2016,Pinkus1999}, it was showed that deep nets outperform shallow
neural networks (shallow nets for short) in the sense that deep nets break through some lower bounds for shallow nets.
Furthermore, some studies \cite{Eldan2015,Sanguineti2013,Lin2017,Mhaskar2016,Raghu2016,Telgarsky2016}
have demonstrated the superiority of deep nets via showing that
deep nets can approximate various functions while shallow nets fail with similar
number of neurons.

Constructing neural networks to approximate continuous functions is a classical and prevalent topic in
approximation theory.
In 1996, Mhaskar \cite{Mhaskar1996} proved that neural networks with single hidden
layer are capable of providing an optimal order of approximating smooth functions.
The problem is, however, that the weights and biases of the constructrs shallow nets are huge,
 which usually leads to extremely large capacity \cite{guo2019}.
 Besides this partly positive approximation results, it was shown in \cite{Chui1996,Lin2017}
 that there is a bottleneck for shallow nets in approximating smooth functions in the
 sense that there is some lower bound for approximation.
Moreover,  Chui et al. \cite{Chui1994} showed
 that shallow nets with an ideal sigmoidal activation
function cannot provide localized approximation in Euclidean space.
Furthermore, it was proved in \cite{Chui2019} that shallow nets cannot capture the
rotation-invariance property by showing the same approximate rates in approximating
rotation-invariant function and general smooth function.
All these results presented limitations of
shallow nets from the approximation theory view point.

To overcome these limitations of shallow nets, Chui et al. \cite{Chui1994} demonstrated that deep nets with two
hidden layers can  provide localized approximation.
Further than that, Chui et al. \cite{Chui2019} showed that deep nets with two hidden layers
and controllable norms of weights
can approximate the univariate smooth functions without saturation and adding depth can realize the rotation-invariance.
Here, saturation \cite{Lin2016} means that the approximation rate cannot be improved once the
smoothness of functions achieves a certain level, which was
proposed as an open question by Chen \cite{chen1993}.
The general results by Lin \cite{Lin2019} indicated that deep nets with two hidden layers and
controllable weights possess both localized and sparse approximation properties in the spatial domain.
They also proved that learning strategies based on deep nets can learn more functions with almost
optimal learning rates than those based on shallow nets.
The problem in \cite{Lin2019} is that the saturation
 cannot be overcome.
The above theoretical verifications demonstrate that deep nets with two hidden layers can
really overcome some deficiency of shallow nets, but that is just partially.

Recent literature in deep nets \cite{Yarotsky2017,Han2019} proved that deep nets with
ReLU activation function (denoting deep ReLU nets) are more efficiently
in approximating smooth function and possess better generalization performance for numerous learning
tasks than shallow nets. Nevertheless, the constructed deep ReLU nets are too deep, which
results in several difficulty in training, including the gradient vanishing phenomenon and disvergence issue \cite{Goodfellow2016}.
Furthermore, how to select the depth is still an open problem, and there is a common phenomenon
that deep nets with huge hidden layers will lead to inoperable \cite{Hinton2006}.
Under this circumstance, we hope to construct a deep net with good approximation capability,
controllable parameters, non-saturation and not too deep.
To this end, we construct in this paper a deep net with three hidden layers that possesses the following
properties: localized approximation, optimal approximation rate,
controllable parameters, non-saturation and spatial sparsity.
Our main tool for analysis is the localized approximation \cite{Chui1996,Lin2019},
``product gate'' strategy \cite{Chui2019,Petersen2018,Yarotsky2017} and localized Taylor polynomials \cite{Han2019,Petersen2018}.


\section{Main results}

Let $I=[0,1]$, $d\in \mathbb{N}$, $x\in X:=I^d$,
$C(\mathbb{R}^d)$ be the space of continuous functions with the norm
$$
\|f\|_{\infty}:=\|f\|_{C(\mathbb{R}^d)}:=\max_{x\in {\mathbb{R}^d}}|f(x)|.
$$
For $x\in X$, the set of shallow nets can be mathematically expressed as
\begin{equation}\label{shallow}
F_{\sigma,n}(x)=\left\{\sum_{i=1}^{n_0}c_i\sigma(w_i\cdot x+b_i): w_i\in\mathbb{R}^d, b_i,c_i\in \mathbb{R}\right\}
\end{equation}
where $\sigma: \mathbb{R}\rightarrow \mathbb{R}$ is an activation function, $n_0$ is the number of hidden neurons (nodes),
$c_i\in \mathbb{R}$ is the outer weights, $w_i:=(w_{ji})_{j=1}^d \in \mathbb{R}^d$
 is the inner weight, and $b_i$ is the bias (threshold) of the $i$-$th$ hidden nodes.

Let $l\in \mathbb{N}$,  $d_0=d, d_1, \cdots, d_l\in \mathbb{N}$,
$\sigma_k: \mathbb{R}\rightarrow \mathbb{R}~(k=1,2,\cdots,l)$ be univariate nonlinear functions.
For $\vec{h}=(h^{(1)},\cdots,h^{(d_k)})^T\in \mathbb{R}^{d_k}$, define
$\vec{\sigma}(\vec{h})=(\vec{\sigma}(h^{(1)}),\cdots,\vec{\sigma}(h^{(d_k)}))^T$.
Denote $\mathcal{H}_{\{\sigma_j, l, \tilde{n}\}}$ as the set of deep nets with $l$ hidden layers and
$\tilde{n}$ free parameters that can be mathematically represented by
\begin{equation}\label{DL-L}
h_{\{\sigma_j, l, \tilde{n}\}}(x)=\vec{a}\cdot\vec{h}_l(x)
\end{equation}
where
$$
\vec{h}_k(x)=\vec{\sigma}_k(W_k\cdot \vec{h}_{k-1}(x)+\vec{b}_k),~k=1,2,\cdots,l,
$$
$h_0(x)=x, \vec{a}\in \mathbb{R}^{d_l}, \vec{b}_k\in \mathbb{R}^{d_k}$,
$W_k:=(W_{i,j}^k)_{d_k\times d_{k-1}}$ is a $d_k\times d_{k-1}$ matrix, and $\tilde{n}$
denotes the number of free parameters, i.e., $\tilde{n}=\sum_{k=1}^l(d_k\cdot d_{k-1}+d_k)+d_l$.
The structure of deep nets, depicted in Figure 1, depends mainly on the structures of the
weight matrices $W_k$ and the parameter vectors $\vec{b}_k$ and $\vec{a}$, $k=1,2,\cdots,l$.
It is easy to see that when $l=1$, the function defined by (\ref{DL-L}) is a shallow net.

%

\begin{figure}[h]
\centering
\includegraphics[height=8cm,width=12.0cm]{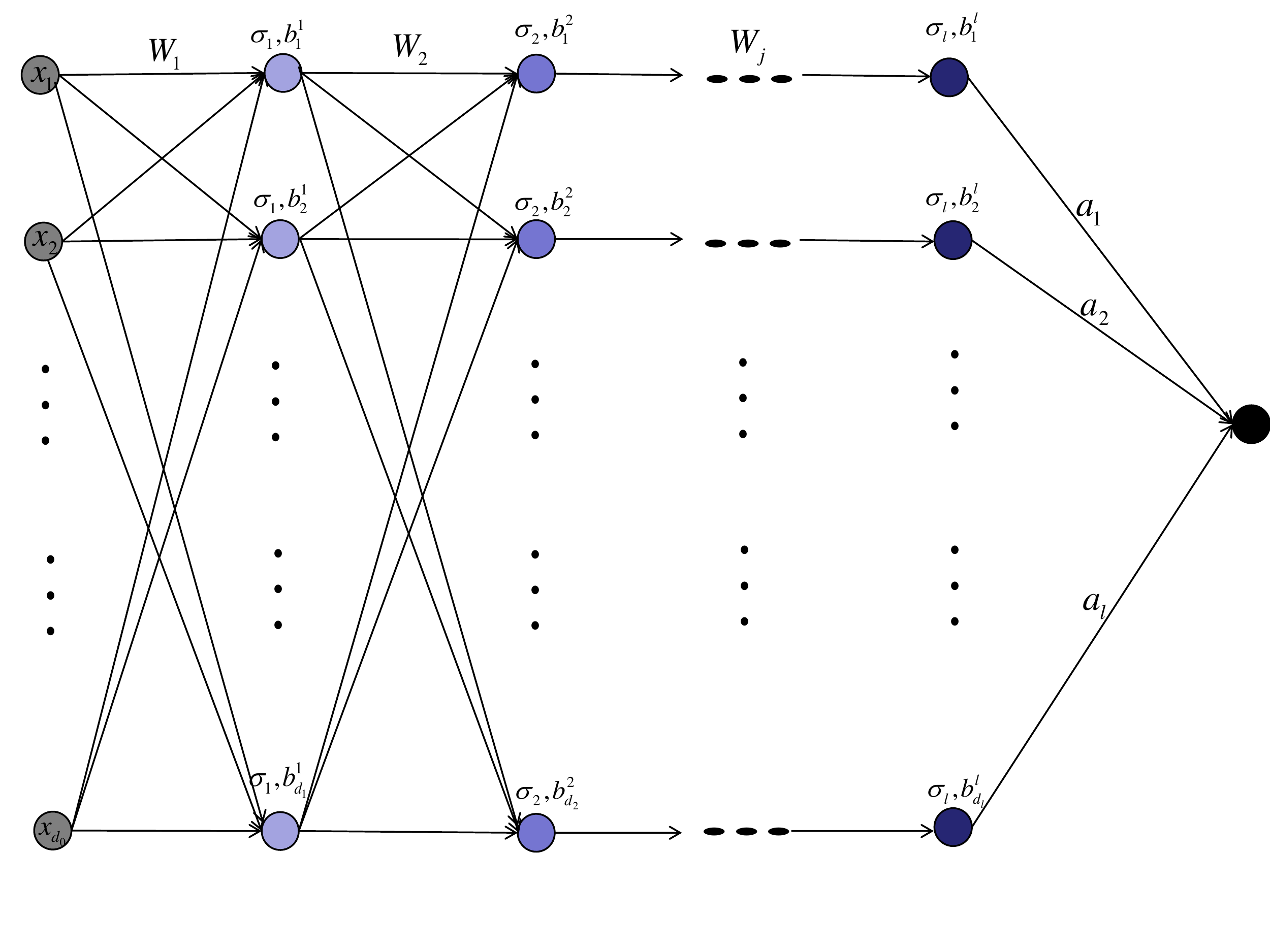}
\caption{Structure for deep neural networks }
\end{figure}

\subsection{Approximation of smooth function by Deep Nets}

In this part, we focus on approximating smooth functions by deep nets.
The smooth property is a widely used priori-assumption
in approximation and learning theory \cite{Chui2018,Cucker2007,Gyorfi2002,Lin2018,Yarotsky2017}.
Let $c_0$ be a positive constant, $r=k+v$ with $k\in \mathbb{N}_0:=\{0\}\bigcup \mathbb{N}$ and $0<v\leq 1$.
A function $f: X\rightarrow \mathbb{R}$ is said to be $(r,c_0)$-smooth if $f$ is $k$-times differentiable and
for any $\alpha_j\in \mathbb{N}_0, j=1,\cdots,d$ with $\alpha_1+\cdots+\alpha_d=k$, then for any $x, z\in X$,
the partial derivatives $\partial^kf/{\partial x_1^{\alpha_1}\cdots \partial x_d^{\alpha_d}}$ exist and satisfy
\begin{equation}\label{smooth}
\left|\frac{\partial^kf}{\partial x_1^{\alpha_1}\cdots \partial x_d^{\alpha_d}}(x)-\frac{\partial^kf}{\partial x_1^{\alpha_1}\cdots \partial x_d^{\alpha_d}}(z)\right|\leq c_0\|x-z\|^v.
\end{equation}
Throughtout this paper, $\|x\|$ denotes the Euclidean norm of $x$.
In particular, if $0<r\leq 1$, then  (\ref{smooth}) coincides the well known Lipschitz condition:
\begin{equation}\label{lip}
|f(x)-f(z)|\leq c_0\|x-z\|^r, \forall x, z\in X.
\end{equation}
Denote by $Lip(r,c_0)$
be the family of $(r,c_0)$-Lipschitz functions satisfying (\ref{lip}).
In fact, the Lipschitz property depicts the smooth information of $f$ and
has been adopted in huge literature \cite{Pinkus1999,Chui1994,Chui2019,Lin2019, Lin2014}
to quantify the approximation ability of neural networks.

As we know, different activation functions used in neural networks will lead to different
results \cite{Pinkus1999}. Among all the activation functions, the sigmoidal function and Heaviside function are two commonly used ones.
Similar as \cite{Lin2019}, we use these two activation functions to construct deep nets.
The main reason is that the usage of Heaviside function can enhance the localized approximation performance \cite{Lin2019}
and the adoption of sigmoidal function can improve the capability to approximate algebraic polynomials \cite{Chui2019}.
Let $\sigma_0$ be the Heaviside function, i.e.,
$$
\sigma_0(t)=1,\mbox{if~} t\geq 0;~\sigma_0(t)=0, \mbox{if~} t<0,
$$
and
$\sigma: \mathbb{R}\rightarrow \mathbb{R}$ be a sigmoidal function, i.e.,
\begin{equation}\label{sigmoid}
\lim_{t\rightarrow +\infty}\sigma(t)=1,\lim_{t\rightarrow -\infty}\sigma(t)=0.
\end{equation}
Due to (\ref{sigmoid}), for any $\varepsilon>0$, there exists a $K_\varepsilon=K(\varepsilon, \sigma)>0$ such that
\begin{eqnarray}\label{k}
\left\{\begin{array}{cc}
|\sigma(t)-1|<\varepsilon,\mbox{~if~}t\geq K_{\varepsilon},\\
|\sigma(t)|<\varepsilon,\mbox{~~if~}t\leq -K_{\varepsilon}.
\end{array}\right.
\end{eqnarray}

Before presenting the main results, we should introduce some assumptions.
Assumption 1 is the $r$-Lipschitz continuous condition for the target function, which is a
standard condition in approximation and learning theory.

{\bf{Assumption 1}}
We assume $g\in \mbox{Lip}(r,c_0)$ with $r=k+v$, $k\in \mathbb{N}_0$, $0<v\leq1$, $c_0>0$.

Assumption 2 concerns the smoothness condition on activation function $\sigma$, which has
already been adopted in \cite{Lin2014almost}.

{\bf{Assumption 2}}
For $r>0$ with $r=k+v$, $k\in \mathbb{N}_0$, $0<v\leq 1$,
let $\sigma$ be a non-decreasing sigmoidal function with
$\|\sigma'\|_{L_\infty(R)}\leq 1$, $\|\sigma\|_{L_\infty(R)}\leq 1$
and there exists at least a point $b_0\in \mathbb{R}^d$ satisfies
$\sigma^{(j)}(b_0)\neq 0$ for all $j=0,1,2,\cdots,k_0$,
and $k_0\geq\max\{k,2\}+1$.

There are many functions satisfy the above restrictions
such as: the Logistic function $\sigma(t)=\frac{1}{1+e^{-t}}$, the Hyperbolic tangent function
$\sigma(t)=\frac{1}{2}(\tanh(t)+1)$, the Gompertz function $\sigma(t)=e^{-ae^{-bt}}$ with $a, b>0$ and the Gaussian function
$\sigma(t)=e^{-t^2}$.

Our first main result is the following theorem, in which we construct a deep net with three
hidden layers to approximate  smooth functions.
Denote by $\mathcal{H}_{3,\tilde{n}}:=\mathcal{H}_{\{\sigma_0,\sigma,\sigma,3,\tilde{n}\}}$ be the set of
deep net with three hidden layers and $\tilde{n}$ free parameters, where $\sigma_0, \sigma, \sigma$ are the activation functions
in the first, second and third hidden layers, respectively.

\begin{theorem}\label{theorem1}
Let $0<\varepsilon\leq1$, under Assumptions 1 and 2,
there exists a deep net $H(x)\in \mathcal{H}_{\{3,\tilde{n}\}}$
such that
\begin{equation}
|g(x)-H(x)|\leq C(\tilde{n}^{-\frac{r}{d}}+\tilde{n}\varepsilon),
\end{equation}
where all  parameters of this deep net are bounded by
$\mbox{poly}(\tilde{n},\frac{1}{\varepsilon})$, $\mbox{poly}(\tilde{n},\frac{1}{\varepsilon})$
denotes some polynomial function with respect to $\tilde{n}$ and $\frac{1}{\varepsilon}$, and $C$ is a constant independent of $\tilde{n}$ and $\varepsilon$.
\end{theorem}

The proof of Theorem \ref{theorem1} will be postponed in Section 4, and a
direct consequences of Theorem \ref{theorem1} is as follows.

\begin{corollary}\label{corollary-1}
Under Assumptions 1 and 2, if $\varepsilon=\tilde{n}^{-\frac{r+d}{d}}$, then there holds
\begin{equation}
|g(x)-H(x)|\leq \bar{C}\tilde{n}^{-\frac{r}{d}},
\end{equation}
where $\bar{C}$ is a constant independent of $\tilde{n}$, and
all the parameters of the deep net are bounded by $\mbox{poly}(\tilde{n})$.
\end{corollary}

%

The approximation rate of shallow nets and
deep nets with two hidden layers
are $\mathcal{O}(\tilde{n}^{-\frac{r}{d}})$ \cite{Maiorov2006,Chui2019},
which is the same as Corollary \ref{corollary-1}.
However, as far as the norm of weights is considered, all the weights in Corollary \ref{corollary-1}
are controllable, and are much less than those of shallow nets.
Specifically, for shallow nets, the norm of weights is at least exponential with
respect to $\tilde{n}$ \cite{Mhaskar1996}, while for deep nets in Corollary \ref{corollary-1}, the norm of
weights is only polynomial respect to $\tilde{n}$.
Such a difference is essentially according to the capacity estimate \cite{guo2019},
where a rigorous proof was presented that the covering number of deep nets with controllable norms of free parameters
can be tightly bounded.  Furthermore, compared with similar results for deep nets with two hidden layers \cite{Lin2019},
we find that our constructed deep net avoids the saturation.
To sum up, the constructed deep net with three hidden layers performs better than
shallow nets and deep nets with two hidden layers in overcoming their shortcomings.

\subsection{Sparse Approximation for Deep Nets}

Sparseness in the spatial domain is a prevalent data feature that abounds in numerous applications such as
magnetic resonance imaging (MRI) analysis \cite{Akkus2017}, handwritten digit recognition \cite{Cire2010} and so on.
The spatial sparseness means that the response (or function) of some actions
only happens on several small regions instead of the whole
input space. In other words, the response vanishes in most of regions of the input space.
Mathematically, the spatially sparse function is defined  as follows \cite{Lin2019}.

Let $s, N\in \mathbb{N}$, $s\leq N^d$,
$N_N^d=\{1,2,...,N\}^d$. Denote by $\{B_{N,\jmath}\}_{\jmath\in N_N^d}$ a cubic
partition of $I^d$ with centers $\{\zeta_\jmath\}_{\jmath\in N_N^d}$
and side length $\frac{1}{N}$. Define
$$
\Lambda_s:=\{\mathbf{k}_{\ell}:\mathbf{k}_{\ell}\in N_N^d,1\leq {\ell} \leq s\}
$$
and
$$
S:=\bigcup_{\jmath\in \Lambda_s}B_{N,\jmath}.
$$

For any function $f$ defined on $I^d$, if the support of $f$ is $S$, then we say that $f$ is $s$-sparse in
$N^d$ partitions.
We use $Lip(N,s,r,c_0)$ to quantify both the smoothness property and sparseness, i.e.,
$$
Lip(N,s,r,c_0)=\left\{f:f\in Lip(r,c_0) ~\mbox{and}~f~\mbox{is s-sparse in}~N^d~\mbox{partition}\right\}.
$$

For $n \in \mathbb{N}$ with $n\geq\hat{c}N$ for some $\hat{c}>0$, let
$\{A_{n,j}\}_{j\in \mathbb{N}_n^d}$ be another cubic partition of $I^d$
with centers $\{\xi_j\}_{j\in \mathbb{N}_n^d}$ and side length $\frac{1}{n}$.
For each $\jmath \in \mathbb{N}_N^d$, define
$$
\bar{\Lambda}_{\jmath}:=\{j\in N_n^d: A_{n,j}\cap B_{N,\jmath}\neq\emptyset\},
$$
it is easy to see that the set $\bigcup_{\jmath\in \bigwedge_s} \bar{\Lambda}_{\jmath}$
is the family of $A_{n,j}$ where $f$ is not vanished.

With these helps, we present a spareness assumption of $f$ as follows.

{\bf{Assumption 3}}
We assume $f\in \mbox{Lip}(N,s,r,c_0)$ with $r=k+v$, $k\in \mathbb{N}_0$,
$0<v\leq1$, $c_0>0$, $N,s \in \mathbb{N}$.

In \cite{Chui2019}, Chui et al. only discussed the approximating performance of deep nets with two hidden layers
in approximating smooth function.
Lin \cite{Lin2019} extended the results in \cite{Chui2019}
to approximate spatially sparse functions.
Specifically, Lin \cite{Lin2019} proved that deep nets with two hidden layers can
approximate spatially sparse function much better than shallow nets. However, their
results suffered from the saturation.
In this subsection, we aim at conquering the above deficiency
by constructing a deep net with three hidden layers.
Theorem \ref{theorem2} below is the second main result of this paper,
and the proof also be verified in Section 4.

\begin{theorem}\label{theorem2}
Let $0<\varepsilon\leq1$, $\tilde{n}\geq \tilde{c}N^d$ for some $\tilde{c}>0$. Under Assumptions 2 and 3,
there exists a deep net $H(x)\in \mathcal{H}_{\{3,\tilde{n}\}}$
such that
\begin{equation*}
|f(x)-H(x)|\leq c_0\tilde{n}^{-\frac{r}{d}}+\tilde{C}\tilde{n}\varepsilon,~\forall x\in X.
\end{equation*}
If $ x\in {I^d} \setminus \bigcup_{\jmath\in \Lambda_s} \bar{\Lambda}_\jmath$, then
\begin{equation*}
|H(x)|\leq \tilde{C}\tilde{n}\varepsilon
\end{equation*}
where all the parameters of the deep net are bounded by
$\mbox{poly}(\tilde{n},\frac{1}{\varepsilon})$, and
$\tilde{C}$ is a constant independent of $\tilde{n}$ and $\varepsilon$.
\end{theorem}

\begin{corollary}\label{corollary-2}
Let $T$ be arbitrary positive number satisfies $T\geq \frac{r+d}{d}$ and $\varepsilon=\tilde{n}^{-T}$.
Under the Assumptions 2 and 3, if $\tilde{n}\geq \tilde{c}N^d$ for some $\tilde{c}>0$,
then there holds
\begin{equation}\label{sparse-1}
|f(x)-H(x)|\leq \hat{C} \tilde{n}^{-\frac{r}{d}},~\forall x\in X.
\end{equation}
If $ x\in X \setminus \bigcup_{\jmath\in \Lambda_s} \bar{\Lambda}_\jmath$, then
\begin{equation}\label{sparse-2}
|H(x)|\leq \tilde{n}^{-T}
\end{equation}
where $\hat{C}$ is a constant independent of $\tilde{n}$.
\end{corollary}

To be detailed, (\ref{sparse-1}) shows that the approximation rate of deep nets
is as fast as $\mathcal{O}(\tilde{n}^{-\frac{r}{d}})$, and (\ref{sparse-2}) states
their performance in realizing the spatial sparseness, when $T$ is large.
However, too large $T$ may lead to extremely large weights, which implies huge
capacity measured by the covering number of $\mathcal{H}_{\{3,\tilde{n}\}}$
according to \cite{guo2019}. A preferable choice of $T$ should be
$T= \mathcal{O}(\frac{r+d}{d})$.


Previous studies \cite{Blum1991, Chui1994} indicated that shallow nets cannot provide localized approximation,
which is a special case for sparse approximation with $s=1$.
Lemma \ref{3.4} (in Section 4) shows that deep nets with two hidden layers
have the localized approximation property, which is the building-block to construct deep
nets possessing sparse approximation property. To the best our knowledge,  \cite{Lin2019} is the first work
 to construct deep nets to realize sparse features.
Compared with \cite{Lin2019}, our main novelty is to deepen the network to conquer the saturation.

\section{Related work}

Constructing neural networks to approximate the functions is a classic
problem \cite{Maiorov2006,Mhaskar1996,Mhaskar2016,Petersen2018,Cucker2007,Gyorfi2002,Lin2019}
in approximation theory.
Traditional method to deal with this problem can be divided into three steps.
Step 1, constructing a neural network to approximate polynomials; Step 2,
utilizing polynomials to approximate target functions;
Step 3, combining the above two steps to reach
the final approximation results between neural networks and target functions.
Tayor formula is usually be used in Step 1 to obtain the approximation results,
which usually leads to extremely large weights, i.e., $|w_i|\sim e^m$, where
$m$ is the degree of the polynomial. However, larger weight leads to large capability and consequently
bad generalization and instable algorithms.
Typical example includes \cite{Mhaskar1996} and \cite{Maiorov1999}.
In order to overcome this drawback, we introduce a new function by the product of Taylor polynomial and a
deep net with two hidden layers  to instead of the polynomial
 in Step 1 to reduce the weights of neural networks from $e^m$ to $\mbox{poly}(m)$.

For deep nets, \cite{Yarotsky2017} and \cite{Petersen2018} stated that
deep ReLU networks are more efficient to approximate
smooth functions than shallow nets. But their results are slightly worse than Theorem \ref{theorem1}
in this paper, in the sense that there is either an additional logarithmic term or under the
weaker norm.
Recently, Han et al. \cite{Han2019} indicated that deep ReLU nets can achieve the optimal generalization
performance for numerous learning tasks, but the depth of \cite{Petersen2018} is much larger than ours.
Recently, Zhou \cite{Zhou2018A,Zhou2018B} also verified that deep convolutional neural network (DCNN) is universal, i.e.,
DCNN can be used to approximate any continuous function to an arbitrary accuracy when the depth of
the neural network is large enough.

All the above literature \cite{Yarotsky2017,Petersen2018,Han2019,Zhou2018A,Zhou2018B}
demonstrated that deep nets with ReLU activation function and DCNN have good properties both in approximation and generalization. However,
there are too deep to be particularly used in real tasks.
Compared with these results, we constructed a deep net only with three hidden layers to approximate
smooth and sparse functions, respectively. We proved in Theorem \ref{theorem1} and
Theorem \ref{theorem2}  that the constructed deep net with three hidden layers and
with controllable weights, can realize smoothness and spatial sparseness without saturation, simultaneously.

\section{Proofs}

Let $\mathcal{P}_m=\mathcal{P}_m(\mathbb{R}^d)$ be the set of  multivariate algebraical polynomials on $\mathbb{R}^d$ of degree at most $m$, i.e.,
$$
{\mathcal{P}}_m=span\{x^k\equiv x_1^{k_1}\cdots x_d^{k_d}:|k|=k_1+...+k_d\leq m \}.
$$
Consider $\mathcal{P}_m^h$ as the set of  homogeneous polynomials of degree $m$, i.e.,
$$
\mathcal{P}_m^h=span\{x^k=x_1^{k_1}\cdots x_d^{k_d}:|k|=k_1+...+k_d=m\}.
$$
%

\subsection{Localized Approximation for Deep Nets}

Let $\mathbb{N}_n^d=\{1,2,\cdots,n\}^d$, $n\in \mathbb{N}$,
$\{A_{n,j}\}_{j\in \mathbb{N}_n^d}$ be the cubic partition of $X$ with centers
$\{\xi_j\}_{j\in \mathbb{N}_n^d}$ and side length $\frac{1}{n}$.
If $x$ lies on the boundary of some $A_{n,j}$, then $j_x$ is the set to be the smallest
integer satisfying $x\in A_{n,j}$, i.e.,
$$
j_x=\left\{j|j\in  N_n^d, x\in A_{n,j}\right\}.
$$
Then, for $K>0$, any $j\in \mathbb{N}_n^d$, $x\in X$, we construct a deep net
with two hidden-layer
$N^*_{n,j,K}(x) \in \mathcal{H}_{\{\sigma_0,\sigma,2,2d+1\}}$ as
\begin{eqnarray}\label{N*}
N^*_{n,j,K}(x)=\sigma \left\{2K\left[\sum_{l=1}^d\sigma_0\left[\frac{1}{2n}+x^{(l)}-\xi_j^{(l)}\right]+
\sum_{l=1}^d\sigma_0 \left[\frac{1}{2n}-x^{(l)}+\xi_j^{(l)}\right]-2d+\frac{1}{2}\right]\right\}. \nonumber \\
\end{eqnarray}

Localized approximation of neural networks \cite{Chui1994}
implies that if the target function is modified only on a small subset
of the Euclidean space, then only a few neurons, rather than the entire network,
need to be retrained.
Lemma {\ref{3.4}} below that was proved in \cite{Lin2019} states the localized approximation property of deep nets
which is totally different from the shallow nets. We refer \cite{Chui2019}
(section 3.3) for details in the localized approximation of neural networks.

\begin{lemma}\label{3.4}
For any $\varepsilon >0$, if $N^*_{n,j,K_\varepsilon}$ is defined by (\ref{N*})
with $K_\varepsilon$ satisfying (\ref{k}) and $\sigma$ being a nondecreasing
sigmoidal function, then

(i) For any $x \not\in A_{n,j}$, there holds $|N^*_{n,j,K_\varepsilon}(x)|<\varepsilon$;

(ii) For any $x\in A_{n,j}$, there holds $|1-N^*_{n,j,K_\varepsilon}(x)|\leq\varepsilon$.
\end{lemma}

It is easy to see that if $\varepsilon\rightarrow 0$, then
$N^*_{n,j,K_\varepsilon}$ is an indicator function for $A_{n,j}$.
Moreover, when $n\rightarrow \infty$, it indicates that $N^*_{n,j,K_\varepsilon}$ can recognize the
location of $x$ in an arbitrarily small region and will vanish in some of
partitions of the input space.

In order to overcome the deficiency of traditional method in neural networks approximation.
We defined a new function $\Phi_g(x)$ by a product of Taylor polynomial and a
deep network function with two hidden layers to instead of polynomials:
\begin{eqnarray}\label{fai}
\Phi_g(x)=\sum_{j\in \mathbb{N}_n^d}P_{k,\eta_{j},g}(x)N^*_{n,j,K_\varepsilon}(x),~x\in X,
\end{eqnarray}
where $N^*_{n,j,K_\varepsilon}$ and $K_\varepsilon$ are defined by (\ref{N*}) and (\ref{k}).
$P_{k,\eta_{j},g}(x)$ is the Taylor polynomial of $g$ with degree $k$
around $\eta_{j}$, $\eta_j \in A_{n,j}$ and $j\in N_n^d$.

Based on the localized approximation results and the localized Taylor polynomial in (\ref{fai}), we construct a deep net
with three hidden layers to approximate both smooth and sparse functions.

\subsection{Proof of Theorem \ref{theorem1}}

The following proposition  indicates that constructing a shallow
net with one neuron can replace a minimal.

\begin{proposition}\label{3.1}
Let $c_0>0, L=\mathcal{O}(m^{d-1})$,
$\sigma\in Lip(r_0,c_0)$ is a sigmoidal function with $(m+1)-$times
bounded derivatives, and
$\parallel {\sigma}' \parallel_{L_{\infty}(\mathbb{R})}\leq 1$,
$\sigma^{(j)}(0)\neq 0$ for all $j=0,1,...,s_0$.
For arbitrary $P_m\in \mathcal{P}_m$ and any $\varepsilon \in (0,1)$, we have
\begin{equation}\label{proposition3.1}
\left| P_m(x)-\sum_{i=1}^L\frac{C(i,m)m!}{\delta_m^m\sigma^{(m)}(0)}\sigma(\delta_m(w_i\cdot x))-P^*_{m-1}(x)\right|<\varepsilon.
\end{equation}
where
$$
P_{m-1}^*(x)=\sum_{j=0}^{m-1}\sum_{i=1}^LD(i,j)(w_i\cdot x)^j,~
D(i,j)=C(i,j)-\frac{C(i,m)m!\sigma^{(j)}(0)}{\delta_m^m\sigma^{(m)}(0)j!},$$
$$\delta_m=\min\left\{\frac{\varepsilon}{M_m},\frac{\varepsilon}{\sum_{i=1}^L|C(i,m)|M_m}\right\},~
M_m=\max_{-1\leq \xi\leq1} \frac{\sigma^{(m+1)}(\xi)}{\sigma^{(m)}(0)},$$
and $C(i,j)$ is an absolute constant.
\end{proposition}

We use the Taylor formula \cite{Chui2019} to prove the above Proposition \ref{3.1}.

\begin{lemma}\label{lemma3.2}
Let $k\geq 1$ and $\varphi$ be $k-$times differentiable on ${\mathbb R}$. Then for any
$u,u_0 \in {\mathbb R}$, there holds
\begin{equation}\label{proposition3.1}
\varphi(u)=\varphi(x_0)+\frac{\varphi^{'}(u_0)}{1!}(u-u_0)+...+\frac{\varphi^{k}(u_0)}{k!}(u-u_0)^k+R_k(u)
\end{equation}
where
\begin{equation}\label{proposition3.1R1}
R_k(u)=\frac{1}{(k-1)!}\int_{u_0}^u[\varphi^{(k)}(t)-\varphi^{(k)}(u_0)](u-t)^{k-1}dt.
\end{equation}
\end{lemma}

In addition, under the condition of Lemma \ref{lemma3.2}, for any $a\in { \mathbb{R}}$, there holds
\begin{equation}\label{proposition3.1-tui}
\varphi(au)=\varphi(u_0)+\frac{\varphi^{'}(u_0)}{1!}(au-u_0)+...+\frac{\varphi^{k}(u_0)}{k!}(au-u_0)^k+R_k(u)
\end{equation}
where
\begin{equation}\label{proposition3.1R2}
R_k(u)=\frac{a^k}{(k-1)!}\int_{u_0}^u[\varphi^{k}(at)-\varphi^{k}(u_0)](u-t)^{k-1}dt.
\end{equation}

Lemma \ref{lemma3.3} which was proved in \cite{Maiorov1999} plays an important role in proving Proposition \ref{3.1}.

\begin{lemma}\label{lemma3.3}
Let $m\in \mathbb{N}$ and $L=\mathcal{O}(m^{d-1})$. For any $P_m \in {\mathcal{P}}_m$, there exists a set of
points $\{w_1,...,w_L\}\subset I^d$ such that
\begin{equation}\label{lemma3.3-1}
P_m(x)=span \{(w_i\cdot x)^j: x, w_i\in I^d, 0\leq j\leq m, 1\leq i\leq L \}.
\end{equation}
\end{lemma}

\begin{proof}[Proof of Proposition \ref{3.1}]
For any $t\in[0,1]$, ${\delta_k}\in(0,1)$, $k\in [1,s_0]$, it follows from Lemma \ref{lemma3.2} that
\begin{equation}\label{proof1}
\sigma({\delta_k} t)=\sigma(0)+\frac{\sigma^{'}(0)}{1!}{\delta_k} t+...+\frac{\sigma^{k}(0)}{k!}({\delta_k} t)^k+\tilde{R}_k(t)
\end{equation}
where
\begin{equation}\label{proof2}
\tilde{R}_k(t)=\frac{\delta_k^k}{(k-1)!}\int_{0}^t[\sigma^{(k)}({\delta_k} u)-\sigma^{(k)}(0)](t-u)^{k-1}du.
\end{equation}
Denote
$$
Q_{k-1}(t)=\sum_{j=0}^{k-1}\frac{\sigma^{(j)}(0)}{j!}({\delta_k} t)^j.
$$
Then (\ref{proof1}) yields
$$
t^k=\frac{k!}{\delta_k^k \sigma^{(k)}(0)}\sigma({\delta_k} t)-\frac{k!}{\delta_k^k\sigma^{(k)}(0)}Q_{k-1}(t)-\frac{k!}{\delta_k^k\sigma^{(k)}(0)}\tilde{R}_k(t),
$$
which implies
\begin{equation}\label{proof3}
t^k=\frac{k!}{\delta_k^k\sigma^{(k)}(0)}\sigma({\delta_k} t)+q_{k-1}(t)+r_k(t)
\end{equation}
where
\begin{equation*}
q_{k-1}(t)=-\frac{k!}{\delta_k^k\sigma^{(k)}(0)}Q_{k-1}(t)
\end{equation*}
and
\begin{equation}\label{rk}
r_k(t)=-\frac{k!}{\delta_k^k\sigma^{(k)}(0)}\tilde{R}_k(t).
\end{equation}
Since
$$
|\sigma^{(k)}({{\delta_k}}u)-\sigma^{(k)}(0)|\leq\max_{0<\xi<1}|\sigma^{(k+1)}(\xi)||{\delta_k}||u|,
$$
then by (\ref{proof2}) and (\ref{rk}), there holds
\begin{eqnarray}\label{R22}
|r_k(t)|&=&\left|-\frac{k!}{\delta_k^k\sigma^{(k)}(0)}\tilde{R}_k(t)\right| \nonumber \\
&=&\frac{k}{|\sigma^{(k)}(0)|}\left|\int_{0}^t[\sigma^{(k)}({\delta_k} u)-\sigma^{(k)}(0)](t-u)^{k-1}du\right| \nonumber \\
&\leq& \frac{k|\sigma^{(k+1)}(\xi)|}{|\sigma^{(k)}(0)|}\left|\int_{0}^t {\delta_k} u(t-u)^{k-1}du\right|=\frac{{\delta_k} k|\sigma^{(k+1)}(\xi)|}{|\sigma^{(k)}(0)|}\left|\int_{0}^1 u(1-u)^{k-1}du\right| \nonumber \\
&\leq& k {\delta_k} M_k(\frac{1}{k}-\frac{1}{1+k})
\leq {\delta_k} M_k.
\end{eqnarray}
From Lemma \ref{lemma3.3}, for any $P_m\in \mathcal{P}_m$, it follows
\begin{eqnarray}\label{P}
P_m(x)&=&\sum_{j=0}^m\sum_{i=1}^L C(i,j)(w_i\cdot x)^j\nonumber \\
&=&\sum_{i=1}^L C(i,m)(w_i\cdot x)^m+\sum_{i=1}^L C(i,m-1)(w_i\cdot x)^{m-1}\nonumber \\
&&+\cdots+\sum_{i=1}^L C(i,1)(w_i\cdot x)+\sum_{i=1}^L C(i,0).
\end{eqnarray}
Since $x, w_i\in I^d$, we have $|w_i\cdot x|\leq 1$.
Then, for an arbitrary $\varepsilon>0$, there exists a $\delta_m\in(0, 1)$
such that
\begin{eqnarray}\label{C}
\sum_{i=1}^L |C(i,m)|M_m\delta_m\leq \varepsilon.
\end{eqnarray}
Due to (\ref{proof3}) and $\delta_m|w_i\cdot x|\in [0,1]$, there holds
\begin{eqnarray}\label{mi}
(w_i\cdot x)^m=\frac{m!}{\delta_m^m\sigma^{(m)}(0)}\sigma(\delta_m(w_i\cdot x))+q_{m-1}(w_i\cdot x)+r_m(w_i\cdot x).
\end{eqnarray}
Inserting the above (\ref{mi}) into (\ref{P}), we obtain
\begin{eqnarray}\label{PP}
P_m(x)&=&\sum_{i=1}^L C(i,m)\left(\frac{m!}{\delta_m^m\sigma^{(m)}(0)}\sigma(\delta_m(w_i\cdot x))
+q_{m-1}(w_i\cdot x)+r_m(w_i\cdot x)\right)\nonumber \\
&&+\sum_{i=1}^L C(i,m-1)(w_i\cdot x)^{m-1}+\cdots+\sum_{i=1}^L C(i,1)(w_i\cdot x)+\sum_{i=1}^L C(i,0)\nonumber \\
&=&\sum_{i=1}^L C(i,m)\frac{m!}{\delta_m^m\sigma^{(m)}(0)}\sigma(\delta_m(w_i\cdot x))+P_{m-1}^*(x)+R_m(x)
\end{eqnarray}
where
$$
R_m(x)=\sum_{i=1}^LC(i,m)r_m(w_i\cdot x)
$$
and
\begin{eqnarray}
P_{m-1}^*(x)&=& \sum_{i=1}^L C(i,m)q_{m-1}(w_i\cdot x)+\sum_{i=1}^L C(i,m-1)(w_i\cdot x)^{m-1}\nonumber \\
&&+\cdots+\sum_{i=1}^L C(i,1)(w_i\cdot x)+\sum_{i=1}^L C(i,0)\nonumber \\
&=& \sum_{i=1}^L D(i,m-1)(w_i\cdot x)^{m-1}+\sum_{i=1}^L D(i,m-2)(w_i\cdot x)^{m-2}\nonumber \\
&&+\cdots+\sum_{i=1}^L D(i,1)(w_i\cdot x)+\sum_{i=1}^L D(i,0)\nonumber \\
&=&\sum_{j=0}^{m-1}\sum_{i=1}^LD(i,j)(w_i\cdot x)^j
\end{eqnarray}
with $D(i,j)=C(i,j)-\frac{C(i,m)m!\sigma^{(j)}(0)}{\delta_m^m\sigma^{(m)}(0)j!}$.
It then follows from (\ref{R22}) and (\ref{C}) that
\begin{eqnarray}\label{RR}
|R_m(x)|\leq \varepsilon.
\end{eqnarray}
Combining (\ref{PP})-(\ref{RR}), we have
$$
\left| P_m(x)-\sum_{i=1}^L\frac{C(i,m)m!}{\delta_m^m\sigma^{(m)}(0)}\sigma(\delta_m(w_i\cdot x+b_i))-P^*_{m-1}(x)\right|<\varepsilon.$$
This completes the proof of Proposition \ref{3.1}.
\end{proof}

Next, we show the performance of shallow nets in approximating.

\begin{proposition}\label{3.2}
Let $\sigma$ be a non-decreasing sigmoidal function with
$\|{\sigma}' \|_{L_{\infty}(\mathbb{R})}\leq 1$,
$\| {\sigma}\|_{L_{\infty}(\mathbb{R})}\leq 1$,
$\sigma^{(j)}(0)\neq 0$ for all $j=0,1,...,m+1$.
For any $P_m\in \mathcal{P}_m$ and $\varepsilon \in (0,1)$, there exists a shallow net
$$
h_{m+1}(x)=\sum_{j=0}^{m}\sum_{i=1}^L a(i,j)\sigma(\delta_j(w_i\cdot x))
$$
with $a(i,j)=\frac{C(i,j)j!}{\delta_j^j\sigma^{(j)}(0)}$ and $\delta_m$ being a polynomial
with respect to $\frac{1}{\varepsilon}$ such that
$$
|P_m(x)-h_{m+1}(x)|\leq\varepsilon.
$$
\end{proposition}

\begin{proof}
From Proposition \ref{3.1}, it holds that
\begin{equation} \label{Pn}
\left| P_m(x)-\sum_{i=1}^La(i,m)\sigma(\delta_m(w_i\cdot x))-P^*_{m-1}(x)\right|<\frac{\varepsilon}{m+1}
\end{equation}
where $a(i,m)=\frac{C(i,m)m!}{\delta_m^m\sigma^{(m)}(0)}$ and $\delta_m\sim\mbox{poly}(\frac{1}{\varepsilon})$.
Similar methods as above
\begin{equation} \label{Pn-1}
\left| P_{m-1}^*(x)-\sum_{i=1}^La(i,m-1)\sigma(\delta_{m-1}(w_i\cdot x))-P^*_{m-2}(x)\right|<\frac{\varepsilon}{m+1},
\end{equation}
$$
\cdots
$$
\begin{equation} \label{P1}
\left| P_{1}^*(x)-\sum_{i=1}^La(i,1)\sigma(\delta_1(w_i\cdot x))-P^*_{0}(x)\right|<\frac{\varepsilon}{m+1},
\end{equation}
and
\begin{equation} \label{P0}
\left|P_{0}^*(x)-\frac{P_{0}^*(x)}{\sigma(0)}\sigma(0\cdot x)\right|=0<\frac{\varepsilon}{m+1}
\end{equation}
Then, it follows from (\ref{Pn})-(\ref{P0}) that
$$
\left| P_m(x)-\sum_{j=0}^m\sum_{i=1}^La(i,m)\sigma(\delta_j(w_i\cdot x))\right|<\varepsilon.
$$
This completes the proof of Proposition \ref{3.2}.
\end{proof}

\begin{corollary}\label{corollary3.2}
Let $m\in \mathbb{N}$, $L=\mathcal{O}(m^{d-1})$ and
$\sigma$ be a non-decreasing sigmoidal function with
$\|{\sigma}' \|_{L_{\infty}(\mathbb{R})}\leq 1$,
$\| {\sigma}\|_{L_{\infty}(\mathbb{R})}\leq 1$.
If  $\sigma^{(j)}(b)\neq 0$ for some $b\in \mathbb{R}$ and all $j=0,1,...,m+1$, then
there exists a shallow net
$$
h_{m+1}(x)=\sum_{j=0}^{m}\sum_{i=1}^L a(i,j)\sigma(\delta_j(w_i\cdot x+b))
$$
such that
$$
|P_m(x)-h_{m+1}(x)|\leq\varepsilon.
$$

\end{corollary}

Based on the above Proposition {\ref{3.2}}, we are able to yield a ``product-gete''
property of deep nets in the following Proposition {\ref{3.3}}, whose proof
can be found in \cite{Chui2019}.

\begin{proposition}\label{3.3}
Let $m\in \mathbb{N}$ and $L=\mathcal{O}(m^{d-1})$. If
$\sigma$ is a non-decreasing sigmoidal function with
$\|{\sigma}' \|_{L_{\infty}(\mathbb{R})}\leq 1$,
$\| {\sigma}\|_{L_{\infty}(\mathbb{R})}\leq 1$,
$\sigma^{(j)}(0)\neq 0$ for all $j=0,1,...,m+1$,
then for any $\varepsilon>0$,
there exists a shallow net
$$
h_3(x)=\sum_{j=1}^3a_j\sigma(w_j\cdot x)
$$
such that for any $u_1,u_2\in [-1,1]$
$$
\left|u_1u_2-\left(2h_3\left(\frac{u_1+u_2}{2}\right)-\frac{1}{2}h_3(u_1)-\frac{1}{2}h_3(u_2)\right)\right|<\varepsilon.
$$
\end{proposition}

\begin{corollary}\label{corollary3.3}
Let $m\in \mathbb{N}$, $L=\mathcal{O}(m^{d-1})$ and
$\sigma$ be a non-decreasing sigmoidal function with
$\|{\sigma}' \|_{L_{\infty}(\mathbb{R})}\leq 1$,
$\| {\sigma}\|_{L_{\infty}(\mathbb{R})}\leq 1$.
If there exists a point $b_0\in \mathbb{R}$
satisfying $\sigma^{(j)}(b_0)\neq 0$ for all $j=1,2,3$,
then for any $\varepsilon>0$,
there exists a shallow net
$$
h_3(x)=\sum_{j=1}^3a_j\sigma(w_j\cdot x+b_0)
$$
such that for any $u_1,u_2\in [-1,1]$
$$
\left|u_1u_2-\left(2h_3\left(\frac{u_1+u_2}{2}\right)-\frac{1}{2}h_3(u_1)-\frac{1}{2}h_3(u_2)\right)\right|<\varepsilon.
$$

\end{corollary}

In our proof, we also need the following Lemma \ref{a2}, which can be found in \cite{Han2019}.

\begin{lemma}\label{a2}
Let $x\in I^d$, $r=k+v$ with $k\in \mathbb{N}_0$ and $0<v\leq 1$. If $f\in Lip(r,c_0)$
and $P_{k,x_0,f}(x)$ is the Taylor polynomial of $f$ with
degree $k$  around $x_0$, then
\begin{equation}\label{lemma 3.5}
|f(x)-P_{k,x_0,f}(x)|\leq \widetilde{c}_1\|x-x_0\|^r, ~x_0\in \mathbb{R}^d,
\end{equation}
where $\widetilde{c}_1$ is a constant depending only on $k,c_0$ and $d$.

\end{lemma}

The following Lemma {\ref{3.5}} illustrates the approximation property of the
 product of Taylor polynomial and deep nets.

\begin{lemma}\label{3.5}
If $g\in \mbox{Lip}(r,c_0)$ with $r=k+v, k\in \mathbb{N}_0$, $0<v\leq1, c_0>0$,
$\sigma$ is a non-decreasing sigmoidal function and
$\Phi(x)$ is defined by (\ref{fai}), then
$$
|g(x)-\Phi_g(x)|\leq \tilde{c_1} n^{-r}+n^dB_0\varepsilon ,~x\in X,
$$
where $B_0:=\|g\|_{L_{\infty}(X)}+\tilde{c_1}$.
\end{lemma}

\begin{proof}
From Lemma {\ref{a2}}, we observe
\begin{equation}\label{bound}
|P_{k,\eta_{j_x},g}(x)| \leq \|g\|_{L_{\infty}(X)}+\tilde{c_1}\|x-\eta_{j_x}\|^r \leq \|g\|_{L_{\infty}(X)}+\tilde{c_1}:=B_0
\end{equation}
where $B_0:=\|g\|_{L_{\infty}(X)}+\tilde{c_1}$.

Since $I^d=\bigcup_{j\in {\mathbb{N}_n^d}}A_{n,j}$,  for each $x\in X$,
there exists a $j_x$ such that $x\in A_{n,j_x}$.
Therefore, it follows from Proposition \ref{3.4} that
\begin{eqnarray}
&&|g(x)-\Phi_g(x)|  \nonumber \\
&=&\left|g(x)-P_{k,\eta_{j_x},g}(x)-\sum_{j\neq j_x}P_{k,\eta_{j},g}(x)N^*_{n,j,K_\varepsilon}(x)
+P_{k,\eta_{j_x},g}(x)(1-N^*_{n,j_x,K_\varepsilon}(x))\right|  \nonumber \\
&\leq& |g(x)-P_{k,\eta_{j_x},g}(x)|+\left|\sum_{j\neq j_x}P_{k,\eta_{j},g}(x)N^*_{n,j,K_\varepsilon}(x)]|
+|P_{k,\eta_{j_x},g}(x)(1-N^*_{n,j_x,K_\varepsilon}(x))\right|   \nonumber \\
&\leq& \tilde{c_1}\|x-\eta_{j_x}\|^r+(n^d-1)B_0\varepsilon+B_0\varepsilon  \nonumber \\
&\leq& \tilde{c_1} n^{-r}+n^dB_0\varepsilon
\end{eqnarray}
This completes the proof of Lemma \ref{3.5}.
\end{proof}

\begin{proof}[Proof of Theorem \ref{theorem1}]
The proof can be divided into three steps:
the first one is to give estimates for the product function and shallow net;
then, we consider the approximation between Taylor polynomial and shallow net;
finally, we give approximation errors by combining the above two steps.

Step 1: By the definition of $N^*_{n,j,K_\varepsilon}(x)$ in (\ref{N*}), we observe
$$
|N^*_{n,j,K_\varepsilon}(x)|\leq 1.
$$
Furthermore, it follows from Lemma {\ref{a2}} that
\begin{equation}\label{bound1}
|P_{k,\eta_j,g}(x)|\leq \|g\|_{L_{\infty}(X)}+\tilde{c_1}\|x-\eta_j\|^r
\leq \|g\|_{L_{\infty}(X)}+\tilde{c_1}.
\end{equation}
Denote
$B_1:=4(\|g\|_{L_{\infty}(X)}+\tilde{c_1}+1)$.
Hence, for an arbitrary $x\in X$, we have $\frac{N^*_{n,j,K_\varepsilon}(x)}{B_1} \in [-\frac{1}{4}, \frac{1}{4}]$
and $\frac{P_{k,\eta_j,g}(x)}{B_1} \in [-\frac{1}{4}, \frac{1}{4}]$.
It then follows from Corollary {\ref{corollary3.3}} with $u_1=\frac{N^*_{n,j,K_\varepsilon}(x)}{B_1}$,
$u_2=\frac{P_{k,\eta_j,g}(x)}{B_1}$ that there exists a shallow net
$$
h_3(x)=\sum_{j=1}^3a_j\sigma(w_j\cdot x+b_j)
$$
such that
\begin{eqnarray}\label{bound3-2}
&&\left|P_{k,\eta_j,g}(x)N^*_{n,j,K_\varepsilon}(x)-B_1^2\left(2h_3\left(\frac{N^*_{n,j,K_\varepsilon}(x)+P_{k,\eta_j,g}(x)}{2B_1}\right)
-\frac{h_3\left(\frac{N^*_{n,j,K_\varepsilon}(x)}{B_1}\right)}{2}-\frac{h_3\left(\frac{P_{k,\eta_j,g}(x)}{B_1}\right)}{2}\right)\right| \nonumber \\
&\leq& B_1^2\varepsilon.
\end{eqnarray}
For the sake of convenience, denote
 $$
 \Delta_1=B_1^2\left(2h_3\left(\frac{N^*_{n,j,K_\varepsilon}(x)+P_{k,\eta_j,g}(x)}{2B_1}\right)-
\frac{h_3\left(\frac{N^*_{n,j,K_\varepsilon}(x)}{B_1}\right)}{2}-\frac{h_3\left(\frac{P_{k,\eta_j,g}(x)}{B_1}\right)}{2}\right).
$$
Noting $\|\sigma^{'}\|_{L_{\infty}(R)}\leq1$, for any $x_1, x_2 \in X$, there holds
\begin{eqnarray}\label{H}
|h_3(x_1)-h_3(x_2)|&\leq& \left|\sum_{j=1}^3a_j\sigma(w_j\cdot x_1+b_j)-\sum_{j=1}^3a_j\sigma(w_j\cdot x_2+b_j)\right|  \nonumber \\
&\leq& \sum_{j=1}^3|a_j||\sigma(w_j\cdot x_1+b_j)-\sigma(w_j\cdot x_2+b_j)|  \nonumber \\
&\leq& \sum_{i=1}^3|a_j|\|w_j\|\|x_1-x_2\| \nonumber \\
&\leq& 3\tilde{c}_2\|x_1-x_2\|,
\end{eqnarray}
where  $\tilde{c}_2=\max_{j\in \{1,2,3\}}|a_j|\max_{i\in \{1,\cdots,d\}}\{|w_{j1}|,\cdots,|w_{jd}|\}$
and $w_j=(w_{j1}$, $\cdots,w_{jd})^T$.

Step 2: It from Corollary {\ref{corollary3.2}} with $P_k(t-\eta_j)=\frac{P_{k,\eta_j,g}(x)}{B_1}$ that
there exists a shallow net
\begin{equation}\label{*}
h_{k+1,L}(x)=\sum_{j=1}^{k+1}\sum_{i=0}^La(i,j)\sigma(w_j\cdot(x-x_i)+b_j)
\end{equation}
such that
$$
\left|\frac{P_{k,\eta_j,g}(x)}{B_1}-h_{k+1,L}(x)\right|\leq \varepsilon_1.
$$

Step 3: Define
\begin{equation}\label{HK}
H(x):=\sum_{j=1}^{n^d}H_j(x)
\end{equation}
where
\begin{equation}\label{***}
H_j(x)=B_1^2\left(2h_3\left(\frac{h_{k+1,L}(x)}{2}+\frac{N^*_{n,j,K_\varepsilon}(x)}{2B_1}\right)-
\frac{h_3(h_{k+1,L}(x))}{2}-\frac{h_3\left(\frac{N^*_{n,j,K_\varepsilon}(x)}{B_1}\right)}{2}\right).
\end{equation}
By (\ref{fai}) and (\ref{bound3-2}), we get
\begin{eqnarray}\label{H-fai-1}
&&|H(x)-\Phi_g(x)|\nonumber \\
&\leq& \sum_{j\in \mathbb{N}_n^d}|H_j(x)-\Delta_1+\Delta_1-P_{k,\eta_j,g}(t)N^*_{n,j,K_\varepsilon}(x)| \nonumber \\
&\leq& \sum_{j\in \mathbb{N}_n^d}\left|B_1^2\left(2h_3\left(\frac{h_{k+1,L}(x)}{2}+\frac{N^*_{n,j,K_\varepsilon}(x)}{2B_1}\right)-
\frac{h_3(h_{k+1,L}(x))}{2}-\frac{h_3\left(\frac{N^*_{n,j,K_\varepsilon}(x)}{B_1}\right)}{2}\right)\right.\nonumber \\
&&\left.-B_1^2\left(2h_3\left(\frac{N^*_{n,j,K_\varepsilon}(x)+P_{k,\eta_j,g}(x)}{2B_1}\right)-
\frac{h_3\left(\frac{N^*_{n,j,K_\varepsilon}(x)}{B_1}\right)}{2}-\frac{h_3\left(\frac{P_{k,\eta_j,g}(x)}{B_1}\right)}{2}\right)\right|
+n^dB_1^2\varepsilon\nonumber \\
&\leq& n^dB_1^2\left(\frac{9}{2}\tilde{c_2}\varepsilon_1+\varepsilon\right) \nonumber \\
&\leq& \bar{C_2}n^d\varepsilon
\end{eqnarray}
where we set $\varepsilon_1=2\varepsilon$ and $\bar{C_2}$ is a constant depending only on $B_1$ and $\tilde{c_2}$.
Noting (\ref{H-fai-1}) and Lemma {\ref{3.5}}, we obtain
\begin{eqnarray*}
|g(x)-H(x)|&\leq& |g(x)-\Phi_g(x)|+|\Phi_g(x)-H(x)| \\
&\leq& \tilde{c_1}n^{-r}+B_0n^d\varepsilon+\bar{C_2}n^d\varepsilon \\
&=& C(n^{-r}+n^d\varepsilon),
\end{eqnarray*}
where $C$ is a constant depending only on $k,c_0,d,B_0$.
Due to (\ref{N*}) (\ref{*}) (\ref{HK}) and (\ref{***}),
there exists a deep net $H(x)\in \mathcal{H}_{3,\tilde{n}}$
with $\tilde{n}=6n^d((L+1)(k+1)+2d+1)$ free parameters satisfying
$$
|g(x)-H(x)|\leq C(\tilde{n}^{-\frac{r}{d}}+\tilde{n}\varepsilon).
$$
Furthermore, it is easy to check (see \cite{Chui2019} for detailed proof)
that all the parameters in $H(x)$ can be bounded by
$\mbox{poly}(\tilde{n},\frac{1}{\varepsilon})$.
This completes the proof of  Theorem \ref{theorem1}.
\end{proof}

\begin{proof}[Proof of Corollary \ref{corollary-1}]
This result can be directly deduced from Theorem \ref{theorem1} with $\varepsilon=\tilde{n}^{-\frac{r+d}{d}}$.
This completes the proof of Corollary \ref{corollary-1}.
\end{proof}

\subsection{Proof of  Theorem \ref{theorem2}}

Since the spatial sparseness depends heavily on the localized approximation property,
we first show that $\Phi_f(x)$ succeeds to realizing the sparseness of the target function $f$
that breaks through the bottleneck of shallow nets \cite{Chui1994}. For different partitions
$\{A_{n,j}\}_{j\in N_n^d}$ and $\{B_{N,j}\}_{j\in N_N^d}$, we
assume $n\geq \hat{c}N$ for some $\hat{c}>0$ throughout the proof.

\begin{lemma}\label{4.1}
Let $0<\varepsilon<1$, under Assumptions 2 and 3, if $\Phi_f(x)$ is defined by (\ref{fai})
and $K_\varepsilon$ satisfies (\ref{k}), then
\begin{equation}\label{proposition4.1-1}
|f(x)-\Phi_f(x)|\leq c_0n^{-r}+B_2 n^d\varepsilon,
\end{equation}
where $B_2=\|f\|_{L_\infty (X)}+\tilde{c}_1$.
Furthermore, if $n\geq \hat{c}N$, there holds
\begin{equation}\label{proposition4.1-2}
|\Phi_f(x)|\leq B_2n^d\varepsilon, ~\forall x\in I^d \setminus \bigcup_{k\in \Lambda_s} \bar{\Lambda}_k.
\end{equation}
\end{lemma}

\begin{proof}
Since $I^d=\bigcup_{j\in N_n^d}A_{n,j}$, for each $x\in X$, there exists a $j_x\in \mathbb{N}$ such that $x\in A_{n,j_x}$.
By Lemma \ref{3.4}, we know that for any $x\in A_{n,j}$, $|1-N_{n,j,K_\varepsilon}^*(x)|\leq\varepsilon$ and
for any $x\not\in A_{n,j}$, $|N_{n,j,K_\varepsilon}^*(x)|\leq\varepsilon$.
From (\ref{bound1}), we also get
\begin{equation}\label{pf}
|P_{k,\eta_j,f}(x)|\leq B_2,
\end{equation}
where $B_2=\|f\|_{L_\infty (X)}+\tilde{c}_1$. Then
\begin{eqnarray*}\label{sparse1}
&&|f(x)-\Phi_f(x)|\nonumber\\
&=&|f(x)-P_{k,\eta_{j_x},f}(x)-\sum_{j\neq j_x}P_{k,\eta_{j},f}(x)N_{n,j,K_\varepsilon}^*(x)+P_{k,\eta_{j_x},f}(x)(1-N_{n,j_x,K_\varepsilon}^*(x))| \\
&\leq&|f(x)-P_{k,\eta_{j_x},f}(x)|+|\sum_{j\neq j_x}P_{k,\eta_{j},f}(x)N_{n,j,K_\varepsilon}^*(x)|+|P_{k,\eta_{j_x},f}(x)(1-N_{n,j_x,K_\varepsilon}^*(x))| \\
&\leq& c_0\|x-\eta_{j_x}\|^r+(n^d-1)B_2\varepsilon+B_2\varepsilon \nonumber\\
&\leq& c_0 n^{-r}+B_2n^d\varepsilon.
\end{eqnarray*}

Since $n>\hat{c}N$, $x\in I^d \setminus \bigcup_{k\in \Lambda_s} \bar{\Lambda}_k$ implies $A_{n,j_x}\bigcap S=\emptyset$.
This together with $f\in Lip(N,s,r,c_0)$ yields $f(x)=P_{k,\eta_{j_x},f}(x)=0$. From Lemma \ref{3.4} and (\ref{pf}), we have
\begin{eqnarray*}\label{sparse2}
|\Phi_f(x)|&=&|\sum_{j\neq j_x}P_{k,\eta_j,f}(x)N_{n,j,K(\varepsilon)}^*(x)|+
|P_{k,\eta_{j_x},f}(x)N_{n,j_x,K(\varepsilon)}^*(x)| \\
&\leq&|P_{k,\eta_{j},f}(x)|\sum_{j\neq j_x}|N_{n,j,K(\varepsilon)}^*(x)|\\
&\leq& (n^d-1)B_2\varepsilon\\
&\leq& B_2 n^d\varepsilon.
\end{eqnarray*}
This completes the proof of Lemma \ref{4.1}.
\end{proof}

\begin{proof}[Proof of Theorem \ref{theorem2}]
The proof of this theorem is similar to the proof of Theorem \ref{theorem1}.
Similar as Step 1 and Step 2 in the proof of Theorem \ref{theorem1}, we obtain
\begin{eqnarray}\label{bound2}
&&\left|P_{k,\eta_j,f}(x)N^*_{n,j,K_\varepsilon}(x)-B_3^2\left(2h_3\left(\frac{2N^*_{n,j,K_\varepsilon}(x)+P_{k,\eta_j,f}(x)}{2B_3}\right)-
\frac{h_3\left(\frac{N^*_{n,j,K_\varepsilon}(x)}{B_3}\right)}{2}-\frac{h_3\left(\frac{P_{k,\eta_j,f}(x)}{B_3}\right)}{2}\right)\right| \nonumber \\
&\leq& B_3^2\varepsilon,
\end{eqnarray}
where $B_3:=2(\|f\|_{L_\infty(X)}+\tilde{c}_1+1)$.
Denote
 $$
 \Delta_2=B_3^2\left(2h_3\left(\frac{2N^*_{n,j,K_\varepsilon}(x)+P_{k,\eta_j,f}(x)}{2B_3}\right)-
\frac{h_3\left(\frac{N^*_{n,j,K_\varepsilon}(x)}{B_3}\right)}{2}-\frac{h_3\left(\frac{P_{k,\eta_j,f}(x)}{B_3}\right)}{2}\right).
$$
Define
$$
H(x):=\sum_{j=1}^{n^d}H_j(x)
$$
where $H_j(x)=B_3^2\left(2h_3\left(\frac{2N^*_{n,j,K_\varepsilon}(x)+h_{k+1,L}(x)}{2B_3}\right)-
\frac{h_3\left(\frac{N^*_{n,j,K_\varepsilon}(x)}{B_3}\right)}{2}-\frac{h_3\left(\frac{h_{k+1,L}(x)}{B_3}\right)}{2}\right)$.\\
Proposition {\ref{3.2}} implies that there exists a shallow net
\begin{equation*}
h_{k+1,L}(x)=\sum_{j=1}^{k+1}\sum_{i=0}^La(i,j)\sigma(w_j\cdot(x-x_i)+b_j)
\end{equation*}
such that
\begin{equation}\label{**}
\left|\frac{P_{k,\eta_j,f}(x)}{B_3}-h_{k+1,L}(x)\right|\leq \varepsilon_1.
\end{equation}
Since $f\in Lip(N,s,r,c_0)$ with (\ref{**}), we obtain
\begin{eqnarray}\label{H-fai}
&&|H(x)-\Phi_f(x)|\nonumber \\
&\leq& \sum_{j\in \mathbb{N}_n^d}|H_j(x)-\Delta_2+\Delta_2-P_{k,\eta_j,f}(x)N^*_{n,j,K_\varepsilon}(x)| \nonumber \\
&\leq& \sum_{j\in \mathbb{N}_n^d}\left|B_3^2\left(2h_3\left(\frac{h_{k+1,L}(x)}{2}+\frac{N^*_{n,j,K_\varepsilon}(x)}{2B_3}\right)-
\frac{h_3(h_{k+1,L}(x))}{2}-\frac{h_3\left(\frac{N^*_{n,j,K_\varepsilon}(x)}{B_3}\right)}{2}\right)\right.\nonumber \\
&&\left.-B_3^2\left(2h_3\left(\frac{N^*_{n,j,K_\varepsilon}(x)+P_{k,\eta_j,f}(x)}{2B_3}\right)-
\frac{h_3\left(\frac{{N^*_{n,j,K_\varepsilon}(x)}}{B_3}\right)}{2}-\frac{h_3\left(\frac{P_{k,\eta_j,f}(x)}{B_3}\right)}{2}\right)\right|
+B_3^2 n^d\varepsilon\nonumber \\
&\leq& \frac{9\tilde{c_2}B_3^2n^d}{2}\left|h_{k+1,L}(x)-\frac{P_{k,\eta_j,f}(x)}{B_3}\right|+B_3^2n^d\varepsilon\nonumber \\
&\leq& \frac{9\tilde{c_2}B_3^2n^d}{2}\varepsilon_1+B_3^2n^d\varepsilon\nonumber \\
&=& B_3^2n^d\varepsilon(9\tilde{c_2}+1)\nonumber \\
&=&\tilde{c}_3n^d\varepsilon
\end{eqnarray}
where $\varepsilon_1=2\varepsilon$, $\tilde{c}_3$ is a constant depending only on $B_3$ and $\tilde{c}_2$.\\
Due to (\ref{H-fai}) and Lemma {\ref{4.1}}, for any $x\in X$, we get
\begin{eqnarray*}
|f(x)-H(x)|&\leq& |f(x)-\Phi_f(x)|+|\Phi_f(x)-H(x)| \\
&\leq& c_0n^{-r}+B_2n^d\varepsilon+\tilde{c}_3n^d\varepsilon \\
&=& c_0n^{-r}+\tilde{C}n^d\varepsilon,
\end{eqnarray*}
where $\tilde{C}$ is a constant depending only on $\tilde{c_3},B_2$.\\
Moreover, if $x\in I^d \setminus \bigcup_{\jmath\in \Lambda_s} \bar{\Lambda}_\jmath$,
we have $f(x)=P_{k,\eta_j,f}(x)=0$ and
$|\Phi_f(x)|\leq B_2n^d\varepsilon$, then it is easy to obtain that
\begin{eqnarray*}
|H(x)|&\leq& |\Phi_f(x)|+|\Phi_f(x)-H(x)| \\
&\leq& |\Phi_f(x))|+|\Phi_f(x)-H(x)|\\
&\leq& B_2n^{d}\varepsilon+\tilde{c}_3n^d\varepsilon\\
&=&\tilde{C}n^d\varepsilon,
\end{eqnarray*}
where $\tilde{C}$ is a constant depending only on $k,c_0,d,B_2$.
It is easy to see that there are totally $\tilde{n}:=6n^d((L+1)(k+1)+2d+1)$ free parameters in $H(x)$.
In this case, we obtain
\begin{equation*}
|f(x)-H(x)|\leq c_0\tilde{n}^{-\frac{r}{d}}+\tilde{C}\tilde{n}\varepsilon.
\end{equation*}
Furthermore, if $ x\in X \setminus \bigcup_{\jmath\in \Lambda_s} \bar{\Lambda}_\jmath$ and $\tilde{n}\geq \tilde{c}N^d$, then
\begin{equation*}
|H(x)|\leq \tilde{C}\tilde{n}\varepsilon.
\end{equation*}
It is noticeable that all the parameters of deep nets are controllable,
which is bounded by $\mbox{poly}(\tilde{n},\frac{1}{\varepsilon})$.
This completes the proof  of Theorem \ref{theorem2}.
\end{proof}

\begin{proof}[Proof of Corollary \ref{corollary-2}]
The result (\ref{sparse-1}) can be  deduced directly from Theorem \ref{theorem2}
with $\varepsilon=\tilde{n}^{-T}$ for $T\geq \frac{r+d}{d}$.
This completes the proof of Corollary \ref{corollary-2}.
\end{proof}

\end{document}